\documentclass{bmcart}
\usepackage{graphicx} 
\usepackage{amsmath,amsfonts,amsbsy,amssymb} 
\usepackage{mathabx} 
\usepackage{amssymb} 
\usepackage{amsmath}
\usepackage{amsthm}	
\usepackage{mathrsfs}
\usepackage[nolist]{acronym} 
\usepackage{tabularx} 
\usepackage{multirow}
\usepackage{wasysym}
\usepackage{float}
\usepackage{cite}
\usepackage{color} 

\usepackage{enumitem} 

\newtheorem{proposition}{Proposition}
\newtheorem{lemma}{Lemma}

\usepackage[utf8]{inputenc} 

\startlocaldefs
\endlocaldefs

\begin{document}

\begin{frontmatter}

\begin{fmbox}
\dochead{Research}


\title{Statistical QoS provisioning for MTC Networks under Finite Blocklength}


\author[
addressref={aff1},                   
corref={aff1},                       
email={mohammad.shehab@oulu.fi}   
]{\inits{MS}\fnm{Mohammad} \snm{Shehab}}
\author[
addressref={aff2},
email={endrit.dosti@mail.utoronto.ca}
]{\inits{ED}\fnm{Endrit} \snm{Dosti}}
\author[
addressref={aff1},
email={hirley.alves@oulu.fi}
]{\inits{HA}\fnm{Hirley} \snm{Alves}}
\author[
addressref={aff1},
email={matti.latva-aho@oulu.fi}
]{\inits{ML}\fnm{Matti} \snm{Latva-aho}}


\address[id=aff1]{
	\orgname{Centre for Wireless Communications (CWC), University of Oulu, Finland} 
}
\address[id=aff2]{%
	\orgname{Department of Electrical and Computer Engineering, University of Toronto, Canada}
}


\begin{artnotes}
\end{artnotes}

\end{fmbox}


\begin{abstractbox}
\begin{abstract}
This paper analyzes the effective capacity of delay constrained machine type communication (MTC) networks operating in the finite blocklength regime. First, we derive a closed-form mathematical approximation for the effective capacity in quasi-static Rayleigh fading channels. We characterize the optimum error probability to maximize the concave effective capacity function with reliability constraint and study the effect of signal-to-interference-plus-noise ratio (SINR) variations for different delay constraints. The trade off between reliability and effective capacity maximization reveals that we can achieve higher reliability with limited sacrifice in effective capacity specially when the number of machines is small. Our analysis reveals that SINR variations have less impact on effective capacity for strict delay constrained networks. We present an exemplary scenario for massive MTC access to analyze the interference effect proposing three methods to restore the effective capacity for a certain node which are power control, graceful degradation of delay constraint and joint compensation. Joint compensation combines both power control and graceful degradation of delay constraint, where we perform maximization of an objective function whose parameters are determined according to delay and SINR priorities. Our results show that networks with stringent delay constraints favor power controlled compensation and compensation is generally performed at higher costs for shorter packets.
\end{abstract}




\end{abstractbox}
%

\end{frontmatter}



\section{Introduction}\label{introduction}
Modern communication systems are becoming an indispensable part of our lives. Driven by the demands of users for extra services, the fifth generation of mobile communication (5G) is expected to introduce new features such as \textit{Ultra Reliable Low Latency Communications (URLLC)} and \textit{massive Machine Type Communication (m-MTC)} \cite {paper1,Dosti,paper3,NokiacMTC2016,paper10,eucnc2}. These features may serve many yet unforeseen applications to enable the Internet of Things (IoT). IoT aims at bringing connectivity to anything that can benefit from internet connection \cite{ericsson}. URLLC has emerged to provide solutions for reliable and low latency transmissions in wireless systems. The design of URLLC systems imposes strict quality of service (QoS) constraints to fulfill very low latency in the order of milliseconds with expected reliability of higher than 99.9$\%$ \cite{paper1,NokiacMTC2016}. In \cite{latency}, Schulz et al. discussed the reliability requirements for different IoT applications. According to their study, latency bounds range from 1 ms in factory automation to 100 ms in road safety. In addition, the packet loss rate constraints range from $10^{-9}$ in printing machines to $10^{-3}$ for traffic efficiency. Such requirements are far more stringent than the ones in the current long term evolution (LTE) standards \cite{Johan}.

The m-MTC refers to networks that can support a variety of connected smart devices at the same time with the same base station (BS). It obligates a certain level of connectivity to a machine via ultra reliable communication (URC) over relatively long term ($>$10 ms) \cite{paper1}. The number of connected devices is expected to cross the 28 billion border by 2021, where a single macro-cell may need to uphold 10,000 or more devices in the future \cite{Hamouda,MTC2}. Moreover, the traffic behavior of MTC is quite different from the HTC (Human type communication), where \cite{mtraffic-issue}:
\begin{itemize}
	\item	MTC is coordinated (i.e, there are simultaneous access attempts from many machine reacting to the same events), while HTC is uncoordinated.
	\item	MTC uses short as well as small number of packets.
	\item	MTC traffic is real as well as non-real time, periodic and event driven.
	\item	MTC QoS requirement is different from HTC (i.e. different reliability and latency requirements).
\end{itemize}

In this context, MTC has gained an increasing interest in recent years via employing new multiple radio access technologies and efficient utilization of spectrum resources to improve reliability and robustness \cite{Orsino2017,eurasip1,eurasip2}. Another research topic that has gathered much attention is cooperative transmission in MTC, where in  \cite{MTC2} the authors proposed a location-based cooperative strategy to reduce the error outage probability, but without study of the latency aspects.

Traditional communication systems are based on Shannon theoretic models and utilize metrics such as channel capacity or ergodic capacity \cite{paper11}. Unlike classical systems, URLLC networks are designed to communicate on short packets in order to satisfy extremely low latency in real time applications and emerging technologies such as e-health, industrial automation, and smart grids whenever data sizes are reasonably small such as sensor readings or alarm notifications. In the finite blocklength regime, the length of metadata is of comparable size with the length of data. Such demands stimulated a revolutionary trend in information theory studying communication at finite blocklength (FB) \cite {paper2,paper5,paper14}. In this context, conventional metrics (e.g. channel capacity or ergodic capacity) become highly suboptimal \cite{paper2}. For this reason, the maximum achievable rate for quasi-static fading channels was characterized in \cite{paper14} as a function of blocklength and error probability $\epsilon$. In \cite{Johan}, the authors analyzed the effect of using smaller resource blocks on error probability bounds in OFDM. The effect of relaying of blocklength-limited packets was studied and compared to direct transmission in \cite{paper12}, \cite{paper13} where the authors concluded that relaying is more efficient than direct transmission in the FB regime specially with average channel state information (CSI). Furthermore, the authors of \cite {paper3} introduced a per-node throughput model for additive white Gaussian noise (AWGN) and quasi-static collision channels. Therein, average delay is considered and interference is treated as AWGN. 

To model the delay requirements in URLLC and MTC networks, we resort to the effective capacity (EC) metric which was introduced in \cite{paper6}. It indicates the maximum possible arrival rate that can be supported by a network with a target delay constraint. In \cite{paper5}, the authors considered quasi-static Rayleigh fading channels and introduced a statistical model for a single node effective rate in bits per channel use (bpcu) for a certain error probability and delay exponent which reflects the latency requirement. However, throughout the paper, a closed form expression for the EC was not provided. Exploiting the EC theory, the authors of \cite{Nokia2} characterized the latency-throughput trade-off for cellular networks .In \cite{paper9}, Musavian et al. analyzed the EC maximization of secondary node with some interference power constraints for primary node in a cognitive radio environment with interference constraints. Three types of constraints were imposed namely average interference power, peak interference power and interference power outage. The fundamental trade-off between EC and consumed power was studied in \cite{paper8} where the authors suggested an algorithm to maximize the EC subject to power constraint for a single node scenario. In \cite{eucnc}, we studied the per-node EC in MTC networks operating in quasi-static Rayleigh fading proposing three methods to alleviate interference namely power control, graceful degradation of delay constraint and the joint method. To the best of our knowledge, EC for FB packets transmission in multi-node MTC scenario has not been investigated until part of the work in this journal was presented in \cite{eucnc}, which will be depicted here with extra details.

Based on its intuition, the EC theory provides a mathematical framework to study the interplay among transmit power, interference, delay, and the achievable rate for different wireless channels. In this paper, we derive a mathematical expression for EC in quasi-static Rayleigh fading for delay constrained networks. Our results depict that a system can achieve higher reliability with a negligible sacrifice in its EC. We consider dense MTC networks and characterize the effect of interference on their EC. We propose three methods to allow a certain node maintain its EC which are: \textit{i}) Power control; \textit{ii}) graceful degradation of delay constraint; and \textit{iii}) joint model. Power control depends on increasing the power of a certain node to recover its EC which in turn degrades the SINR of other nodes. Our analysis proves that SINR variations have limited effect on EC in networks with stringent delay limits. Hence, the side effect of power control is worse for less stringent delay constraints and vice versa. We illustrate the trade off between power control and graceful degradation of delay constraint. Furthermore, we introduce a joint model which combines both of them. The operational point to determine the amount of compensation performed by each of the two methods in the joint model is determined by maximization of an objective function leveraging the network performance. 

The motivation beyond this paper is to provide a solid understanding of the trade off between power, delay, and reliability in MTC networks in the finite blocklength regime. Our objective is to pave the road for utilizing short packets in 5G and machine type networks. Extra plots that were not present in \cite{eucnc} are illustrated to provide the reader with full understanding of the objective function in joint compensation and the compensation process itself. Moreover, we extend the analysis in \cite{eucnc} by solving the optimization problem to obtain the optimum error probability which maximizes the EC in the ultra reliable region. We also characterize the trade off between reliability and EC which shows that we can obtain a huge gain in reliability in return for a slight reduction in EC.

The rest of the paper is organized as follows: in  Section \ref {system model}, we introduce the system model and define some concepts such as communication at finite blocklength and EC. A closed form for the EC in quasi-static Rayleigh fading is derived in Section \ref{EC_FB}, where we also show the effect of interference on the per-node EC in multi-node MTC networks. Next, Section \ref{EC_UR} depicts the optimization problem to maximize the EC in the UR region. We present the interference alleviation methods and the trade off between them in Section \ref{multinode}. Finally, Section \ref{conclusion} concludes the paper. Table 1 includes the important abbreviations and symbols that will appear throughout the paper.

\begin{table}[h!]
	\caption{List of abbreviations and symbols.}
	\begin{tabular}{cccc}
		\hline
		bpcu & bit per channel use\\
		EC & Effective Capacity \\
		max & maximize \\
		PDF & Probability Density Function \\
		QoS & Quality Of Service \\
		SINR & Signal-to-Interference-plus-Noise Ratio \\
		s.t & subject to \\
		URC & Ultra Reliable Communication  \vspace{7mm} \\

		$C(\rho|h|^2)$ & Shannon capacity \\
		$D_{max}$ & maximum delay \\
		$\mathbb{E}[ \ ]$ & expectation of \\
		$EC$ & effective capacity \\
		$EC_{max}$ & maximum effective capacity \\
		$\mathfrak{L(\epsilon,\lambda)}$ & Lagrangian function \\
		$N$ & number of nodes \\
		$Pr()$ & probability of \\
		$P_{out_delay}$ & delay outage probability \\
		$Q(x)$ & Gaussian Q-function \\
		$Q^{-1}(x)$ & inverse Gaussian Q-function \\
		$T_f$ & blocklength \\
		$V(\rho|h|^2)$ & channel dispersion \vspace{7mm} \\

		$e$ & exponential Euler's number \\
		$|h|^2$ & fading coefficient \\
		$\ln$ & natural logarithm to the base $e$ \\
		$\log_2$ & logarithm to the base 2 \\
		$r$ & normalized achievable rate \\
		$\mathbf {w}$ & additive while Gaussian noise vector \\
		$\mathbf {x}_n$ & transmitted signal vector of node n \\
		$\mathbf {y}_n$ & received signal vector of node n \\
		$z$ & fading random variable \vspace{7mm} \\
		
		$\alpha$ & collision loss factor \\ 
		$\alpha_c$ & compensation loss factor \\
		$\alpha_{c_o}$ & operational point of compensation loss factor \\
		$\alpha_t$ & total loss \\
		$\gamma_c$ & compensation gain \\
		$\theta$ & delay exponent \\
		$\epsilon$ & error probability \\
		$\epsilon_t$ & target error probability \\
		$\epsilon^*$ & optimum error probability \\
		$\eta_{\alpha}$ & compensation loss priority factor \\
		$\eta_{\theta}$ & delay priority factor \\
		$\rho$ & signal to noise ratio \\
		$\rho_c$ & compensation SNR \\
		$\rho_{c_o}$ & operational point of compensation SNR\\
		$\rho_i$ &  signal-to-interference-plus-noise ratio \\
		$\rho_s$ & SINR of other non-compensating nodes \\\hline
	\end{tabular}
\end{table}

\section{Preliminaries} \label{system model}
\vspace{-0mm}
\subsection{Network model}
\vspace{-0mm}
We consider a transmission scenario in which $N$ nodes transmit packets with equal power to a common controller through a quasi-static Rayleigh fading collision channel with blocklength $T_f$ as shown in Figure 1. For convenience in this paper, we refer to one machine terminal as node. Given that all nodes transmit at the same time slot, the controller attempts to decode the transmitted symbols arriving from all nodes. When the controller decodes one node's data, the other streams appear as interference to it. For this model, imagine that a node needs to rise its rate temporarily for a critical reason. Later on, we study the interference alleviation scenarios for one node at a certain time slot while all nodes also keep transmitting at the same time.

Based on our network model, the received vector $\mathbf {y}_n\in \mathbb{C}^n$ of node $n$ is given by
\begin{align}\label{eq1}
	\mathbf {y}_n=h_n\mathbf {x}_n+\sum_{s\neq n} h_s\mathbf {x}_s+\mathbf {w},
\end{align}
where $\mathbf {x}_n \in \mathbb{C}^n$ is the transmitted packet of node $n$, $h_n$ is the fading coefficient for node $n$ which is assumed to be quasi-static with Rayleigh distribution. 
This implies that the fading coefficient $h_n$ remains constant for each block of $T_f$ channel uses which span the whole packet duration and changes independently from one block to another. The index $s$ includes all $N-1$ interfering nodes which collide with node $n$, and $\mathbf{w}$ is the additive complex Gaussian noise vector whose entries are defined to be circularly symmetric with unit variance. Given the signal to noise ratio $\rho$ of a single node, the signal-to-interference-plus-noise ratio of any node $n$ is
\begin{align}\label{eq2}
	\rho_i=\frac{\rho}{1+\rho\sum_{s} |h_s|^2}.
\end{align}
\begin{figure}[!t] 
	\centering
	\includegraphics[width=0.98\columnwidth]{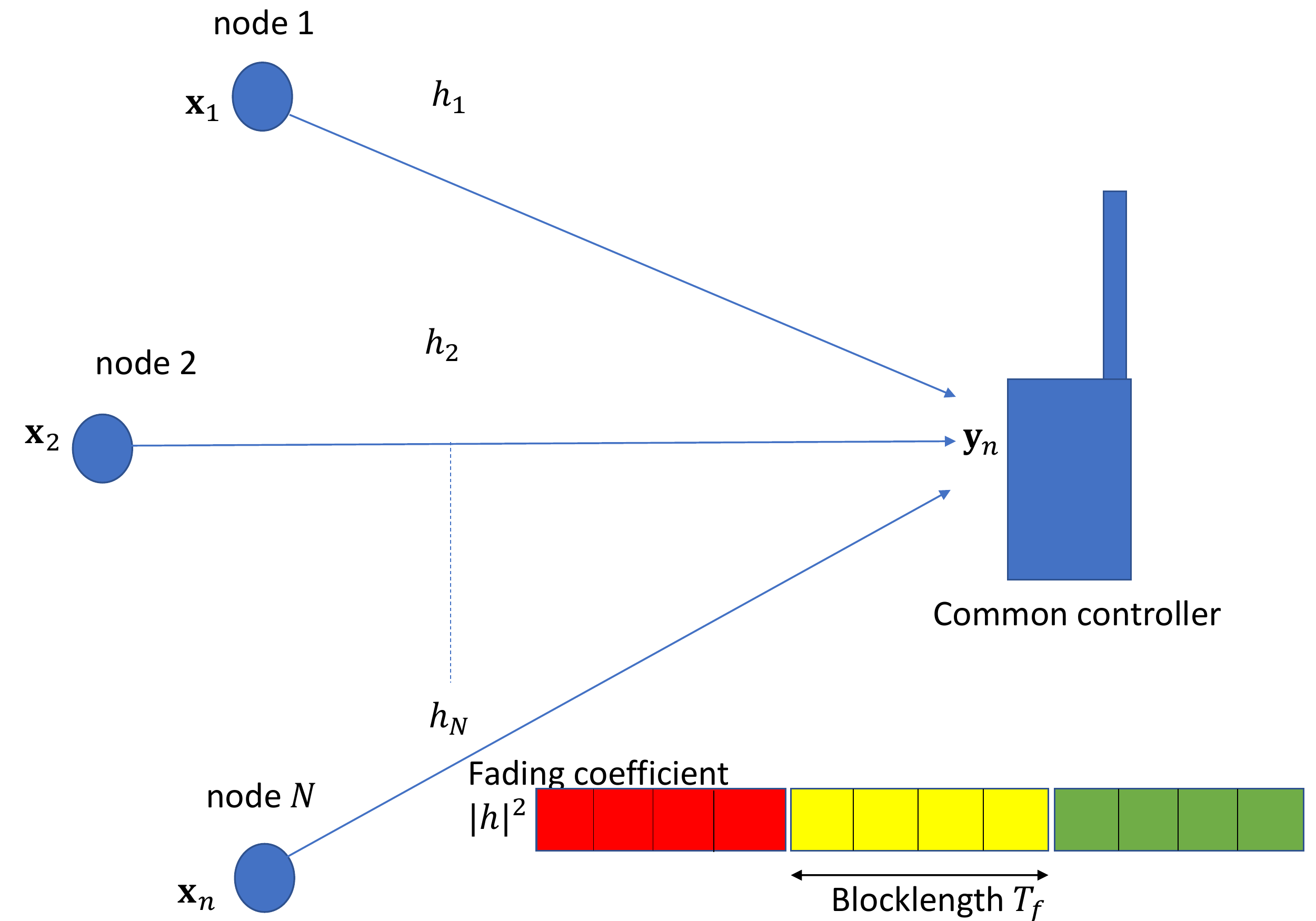}
	\vspace{-0mm}
	\caption{Network Layout.}
	\label{channel}
	\vspace{-0mm}
\end{figure}
To simplify the analysis, we assume that: \textit{i}) each node always has a packet to transmit (buffer is always non-empty); \textit{ii}) all nodes are equidistant from the common controller (i.e., same path loss); and \textit{iii}) the fading coefficients $h_s$ are independent and identically distributed and perfectly known to the receiver. Thus, as the number of nodes increases, the sum of Rayleigh distributed fading envelopes of $N-1$ interfering nodes becomes $\sum_{s} |h_s|^2\approx N-1$ \cite{fundamentals}, and the interference resulting from nodes in set $s$ can be modeled as in \cite{paper3} where (\ref{eq2}) reduces to 
\begin{align}\label{eq2.2}
	\rho_i=\frac{\rho}{1+\rho \ (N-1)}.
\end{align}
Note that CSI acquisition in this setup is not trivial and its cost is negligible whenever the channel remains constant over multiple symbols. Additionally, as in \cite{paper3} we aim to provide a performance benchmark for such networks without interference coordination. 
\vspace{-0mm}
\subsection{Communication at Finite Blocklength}
In this section, we present the notion of FB transmission, in which short packets are conveyed at rate that depends not only on the SNR, but also on the blocklength and the probability of error $\epsilon$ \cite{paper2}. In this case, $\epsilon$ has a small value but not vanishing. For error probability $\epsilon \in\left[ 0,1\right] $, the normalized achievable rate in bpcu is given by 
\begin{align}\label{eq3}
	\begin{split}
		r\approx C\left( \rho_i|h|^2\right) -\sqrt{\tfrac{V(\rho_i |h|^2)}{T_f} }  \operatorname{Q}^{-1}
		(\epsilon), 	
	\end{split}
\end{align}
where
\begin{align}\label{eq3.1}
	C(t)=\log_2(1+t)
\end{align}
is Shannon's channel capacity for sufficiently long packets, while
\begin{align}\label{eq3.2}
	V(t)=\left(1-\left( 1+t\right)^{-2}\right)\left(  \log_2 e\right) ^2
\end{align}
denotes the channel dispersion which appears for relatively short packets ($T_f<2000$ channel uses) \cite{paper5}, $Q(t)=\int_{t}^{\infty}\frac{1}{\sqrt{2 \pi}}e^{\frac{-s^2}{2}} ds$ is the Gaussian Q-function, and $Q^{-1} (t)$ represents its inverse, $\rho_i$ is the SINR and $|h|^2$ is the fading envelope.  

The channel is assumed to be Rayleigh quasi-static fading where the fading coefficients remain constant over $T_f$ symbols which spans the whole packet duration. For Rayleigh channels \cite{alouini}, the fading coefficients $Z=|h|^2$ have the following probability density function distribution 
\begin{align}\label{pdf}
	f_Z(z)=e^{-z}.
\end{align}

\subsection{Effective Capacity}  
The concept of EC indicates the capability of communication nodes to exchange data with maximum rate and certain latency constraint and thus, guarantees QoS by capturing the physical and link layers aspects. A statistical delay violation model implies that an outage occurs when a packet delay exceeds a maximum delay bound $D_{max}$ and its probability is defined as \cite{paper6}
\begin{align}\label{delay}
	P_{out\_ delay}=Pr(delay \geq D_{max}) \approx e^{-\theta \cdot EC \cdot D_{max}},
\end{align}	 
where $\Pr(\cdot)$ denotes the probability of a certain event. Conventionally, a network's tolerance to long delay is measured by the delay exponent $\theta$. The network has more tolerance to large delays for small values of $\theta$ (i.e., $\theta\rightarrow 0$), while for large values of $\theta$, it becomes more delay strict. For the infinite blocklength model, the EC capacity is defined as
\begin{align}\label{EC_Shanon}
	EC(\rho_i,\theta)=-\frac{\ln\mathbb{E}_{Z}\left(e^{-\theta T_f C\left( \rho_i |h|^2\right)} \right)}{T_f\theta},
\end{align} 
In quasi-static fading, the channel remains constant within each transmission period $T_f$ \cite{paper13}, and the EC is subject to the finite blocklength error bounds and thus, according to \cite{paper5} can be written as
\begin{align}\label{EC}
	EC(\rho_i,\theta,\epsilon)=-\frac{\ln\psi(\rho_i,\theta,\epsilon)}{T_f\theta},
\end{align} 
where
\begin{align}\label{psi}
	\psi(\rho_i,\theta,\epsilon)=\mathbb{E}_{Z}\left[\epsilon+(1-\epsilon)e^{-T_f \theta r}\right].
\end{align} 

In [\cite{paper5} and \cite{paper4}, the effective capacity is statistically studied for single node scenario in block fading, but never to a closed form expression. It has been proven that the EC is concave in $\epsilon$ and hence, has a unique maximizer. In what follows, we shall represent the EC expression for quasi-static Rayleigh fading.

\section{Effective Capacity analysis under Finite Blocklength} \label{EC_FB}
\begin{lemma} \label{L1}
	The effective capacity of a certain node communicating in a quasi-static Rayleigh fading channel is approximated by
	\begin{align}\label{Rayleigh}
		\begin{split}
			EC(\rho_i,\theta,\epsilon)\approx-\frac{1}{T_f\theta}  \ln \left[\epsilon +(1-\epsilon) \  \mathcal{J}\right],
		\end{split}
	\end{align}
	with
	\begin{align}\label{J}
		\begin{split}
			\mathcal{J}=e^{\frac{1}{\rho}} \rho^d\left[ \left( \frac{c^2}{2}+c+1\right) \Gamma\left(d+1,\frac{1}{\rho}\right)
			-\frac{c}2{}\left( c+1\right)  \rho^{-2}\Gamma\left(d-1,\frac{1}{\rho}\right)\right] ,
		\end{split}
	\end{align}
	where $d=\frac{-\theta T_f}{\ln(2)}$. Also let $c=\theta \sqrt{T_f} Q^{-1}(\epsilon)\log_2e$ and $x=\sqrt{1-\frac{1}{(1+\rho_i z)^{2}}}$, and $\Gamma(\cdot, \cdot)$ is the upper incomplete gamma function \cite[\S 8.350-2]{Gradshteyn}.
	
\end{lemma}
\begin{proof}.
Please refer to Appendix A.
\end{proof}

\begin{lemma} \label{corollary 1}
	There is a unique local and global maximizer in $\epsilon$ for the per-node EC in quasi-static Rayleigh fading channels which is given by  
	\begin{align}\label{e*}
		\begin{split}
			\epsilon^*(\rho_i,c,d)=\arg\min_{0 \leq \epsilon \leq 1} \psi(\rho_i,c,d)\approx \  \epsilon +(1-\epsilon) \  \mathcal{J}
		\end{split}
	\end{align} 
\end{lemma}
\begin{proof}
	The expectation in (\ref{EC}) is shown to be convex in $\epsilon$ in \cite{paper5} independent of the distribution of channel coefficients $z=|h|^2$. Thus, it has a unique minimizer $\epsilon^*$ which is consequently the EC maximizer given by (\ref{e*})
\end{proof}
Note that $c=\theta \sqrt{T_f} Q^{-1}(\epsilon)\log_2e$ is not a function of $z$, so it can be taken out of the integration which simplifies the optimization problem. To obtain the maximum per-node effective capacity $EC_{max}$, we simply insert $\epsilon^*$ into (\ref{Rayleigh}).

Having obtained the closed-form solution for EC, we proceed with studying the effect of multi-node interference on the per-node EC. We elaborate the effect of interference in quasi-static Rayleigh fading by plotting the per-node EC obtained from Lemma \ref{L1} and (\ref{eq2.2}) for 1, 5 and 10 machines in Figure 2. The network parameters are set as $T_f=1000, \rho=2$, and $\theta=0.01$. It is obvious that the per-node EC decreases when increasing the number of machines $N$ as more interference is added. Notice that the EC curves are concave in $\epsilon$ as envisaged by \cite{paper5} and hence, have a unique maximizer which is obtained from Lemma \ref{corollary 1} and depicted in the figure. Another observation worth mentioning is that the optimum probability of error $\epsilon^*$ which maximizes the EC becomes higher when increasing the number of machines. Notice that Figure 2 confirms that Lemma \ref{corollary 1} renders an accurate approximation to (\ref{EC}).

\begin{figure}[!t] 
	\centering
	\includegraphics[width=0.95\columnwidth]{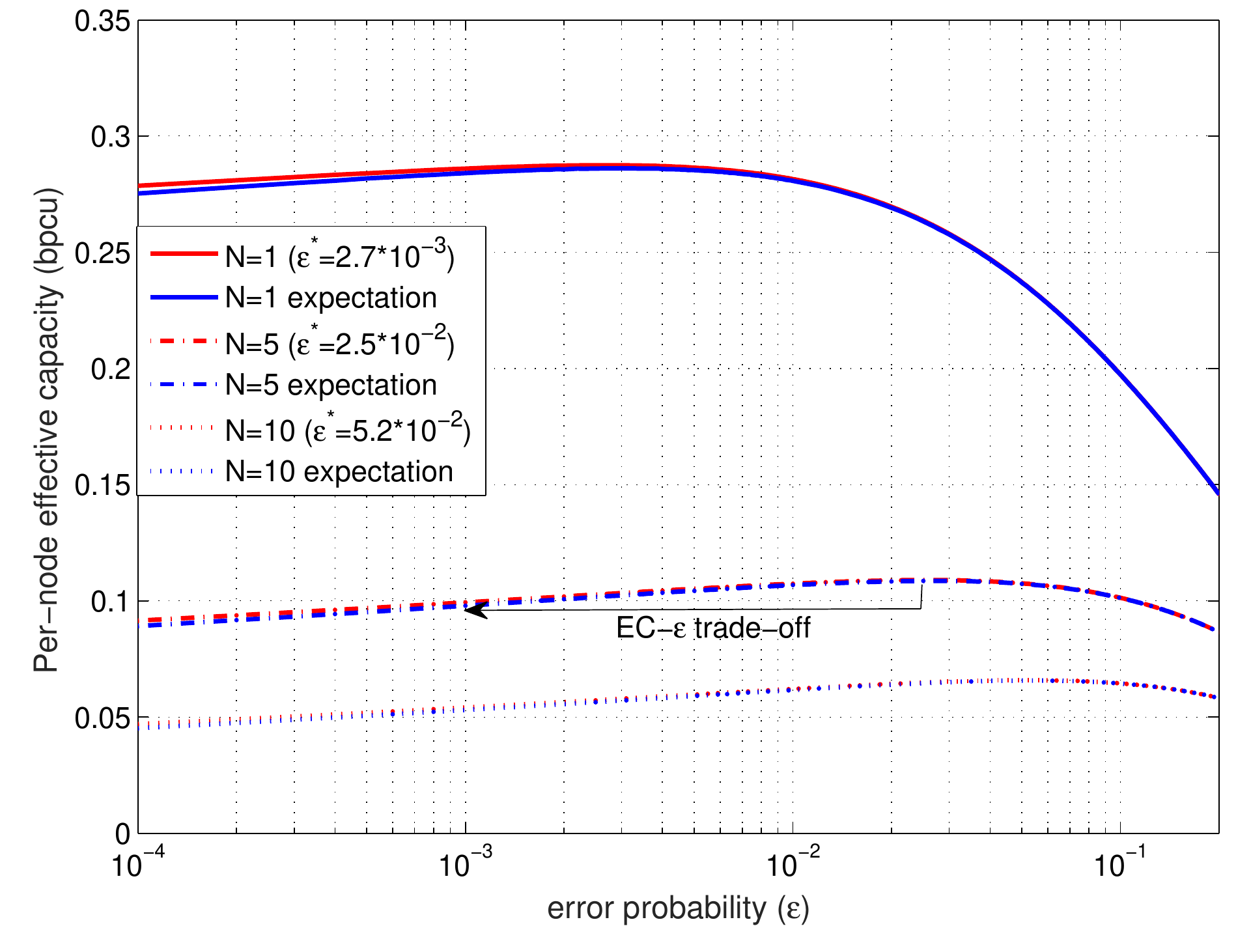}
	\centering
	\vspace{-0mm}
	\caption{Per-node effective capacity in bpcu as a function of error outage probability $\epsilon$ for different number of nodes, with $T_f=1000, \rho=2$, and $\theta=0.01$. }
	\label{interference effect}
	\vspace{-0mm}
\end{figure}

\section{Maximization of effective capacity in the ultra-reliable region} \label{EC_UR}
Taking a closer look at Figure 2, we observe the trade off between the per-node $EC$ and error probability $\epsilon$. It is apparent that we can earn a lower error probability $\epsilon$ by sacrificing only a small amount of EC. For example, we observe that for the 5 nodes network operating in quasi-static Rayleigh fading channel, if we tolerate a decrease in the EC from 0.11 to 0.1 bpcu, the error probability $\epsilon$ improves to $10^{-3}$ instead of $2.5\times10^{-2}$. Thus, sacrificing only 9$\%$ of the EC maximum value boosts the error probability by nearly 1250$\%$ and hence, leads to a dramatic enhancement of reliability. This observation is an open topic for analysis of the EC-$\epsilon$ trade-off with a target of maximizing the EC with some error constraint reflecting the reliability guarantees. 

In this section, we discuss the EC maximization in ultra reliable (UR) region (i.e when the probability of error is extremely low). We maximize the EC so that the error outage probability stays below a very small target value $\epsilon_t$. We define the optimization problem as
\begin{equation}\label{ur1}
\begin{split}
\max_{ } \ &EC(\rho_i,\theta,\epsilon) \\ 
s.t \  \ \ &\epsilon \leq \epsilon_t,
\end{split}
\end{equation} 
which according to (\ref{general2}) in Appendix A can be interpreted to 
\begin{align}\label{ur2}
	\begin{split}
		\min \ &\psi(\epsilon)\approx  \epsilon+(1-\epsilon)\vphantom{\int_{0}^{\infty}} \sum_{n=0}^2\int_{0}^{\infty}(1+\rho_i z)^{d} \frac{(cx)^n}{n!}e^{-z} dz, \\
		s.t \ \ &\epsilon \leq \epsilon_t. 
	\end{split}
\end{align}

The solution of (\ref{ur2}) renders the operational EC which guarantees ultra reliability according to the error constraint. Operational EC in this case is smaller than or equal to the maximum effective capacity $EC_{max}$ obtained from Lemma \ref{corollary 1}. Due to the convexity of $\psi(\epsilon)$ \cite{paper5}, the first derivative of $\psi(\epsilon)$ is positive if $\epsilon$ is greater than the global maximum and vice versa. Thus, we can check if the optimum solution is given by $\epsilon_t$ or not through the first derivative of $\psi(\epsilon)$. To elucidate, we write down the Lagrangian of (\ref{ur2}) as
\begin{align}\label{ur3}
	\mathfrak{L(\epsilon,\lambda)}=\psi(\epsilon)+\lambda(\epsilon-\epsilon_t),
\end{align}
where $\lambda$ is the Lagrangian multiplier. This leads to the following Karush$-$Kuhn$-$Tucker (KKT) conditions \cite{Boyd}
\begin{subequations} \label{ur4}
	\begin{align}
		\frac {\partial \mathfrak{L}}{\partial \epsilon}=\frac {\partial \psi(\epsilon)}{\partial \epsilon}+\lambda=0 \\
		\lambda(\epsilon-\epsilon_t)=0.
	\end{align}
\end{subequations}	
From the second condition, if $\lambda$ is greater than zero, this means that the constraint is active, where $\epsilon^*=\epsilon_t$ and $\frac {\partial \psi(\epsilon)}{\partial \epsilon}|_{\epsilon=\epsilon_t}$ is indeed negative. Reversing this conclusion, we can infer whether the constraint is active or not from the sign of $\frac {\partial \psi(\epsilon)}{\partial \epsilon}|_{\epsilon=\epsilon_t}$ so that
\begin{equation}\label{ur5}
\epsilon^*=\left\{\begin{array}{cl}
\epsilon_t & \frac {\partial \psi(\epsilon)}{\partial \epsilon}|_{\epsilon=\epsilon_t} < 0\\
\arg \min_{\epsilon \geq 0} \ \psi(\epsilon) , & \frac {\partial \psi(\epsilon)}{\partial \epsilon}|_{\epsilon=\epsilon_t} > 0
\end{array}
\right..
\end{equation}
The first derivative of $\psi(\epsilon)$ with respect to $\epsilon$ is derived as follows. From (\ref{general2}), we have 
\begin{align}\label{ur6}
	\begin{split}
		\mathcal{J}&=\sum_{n=0}^2\int_{0}^{\infty}(1+\rho_i z)^{d} \frac{(cx)^n}{n!}e^{-z} dz \\
		&=\mathcal{J}_1+c\mathcal{J}_2+\frac{c^2}{2}\mathcal{J}_3, 
	\end{split}
\end{align}
where
\begin{align}\label{ur8}
	\begin{split}
		\mathcal{J}_1&=e^{\frac{1}{\rho}} \rho^d\Gamma\left(d+1,\frac{1}{\rho}\right) \\
		\mathcal{J}_2&=e^{\frac{1}{\rho}} \rho^d\Gamma\left(d+1,\frac{1}{\rho}\right) -\frac{1}{2} e^{\frac{1}{\rho}} \rho^{d-2}\Gamma\left(d-2,\frac{1}{\rho}\right) , \\
		\mathcal{J}_3&= e^{\frac{1}{\rho}} \rho^d\Gamma\left(d+1,\frac{1}{\rho}\right) - e^{\frac{1}{\rho}} \rho^{d-2}\Gamma\left(d-2,\frac{1}{\rho}\right), \\
	\end{split}
\end{align}

Let $\delta=\theta\sqrt{T_f}\log_2e$, then
\begin{align}\label{ur7}
	\begin{split}
		\frac {\partial \psi(\epsilon)}{\partial \epsilon}&=1+(1-\epsilon)\frac {\partial \mathcal{J}}{\partial \epsilon}-\mathcal{J}\\
		&=1+(1-\epsilon)\left(\delta q(\epsilon) \mathcal{J}_2+\delta^2 Q^{-1}(\epsilon)q(\epsilon)\mathcal{J}_3 \right)-\mathcal{J}, 
	\end{split}
\end{align}
and $q(\epsilon) = \frac{\partial Q^{-1} (\epsilon)}{\partial \epsilon}=-\sqrt{2\pi} e^{\frac{Q^{-1}(\epsilon)}{2}}$ is the first derivative of $Q^{-1}(\epsilon)$ w.r.t $\epsilon$. Substituting, we get 
\begin{align}\label{ur8}
	\begin{split}
		\frac {\partial \psi(\epsilon)}{\partial \epsilon}=1-(1-\epsilon)\delta \sqrt{2\pi}e^{\frac{(Q^{-1}(\epsilon))^{2}}{2}}(\mathcal{J}_2+\delta \mathcal{J}_3)-\mathcal{J}.
	\end{split}
\end{align}

\section{Methods} \label{multinode}

Given that all nodes transmit at the same time slot, the controller attempts decoding the transmitted symbols arriving from all of them. When the controller decodes one node's data, the other streams appear as interference to it \cite{paper3}. For this model, imagine that a node needs to raise its EC in order to meet its QoS constraint. At the first glance, applying successive interference cancellation at the base station would seem to be attractive solution. However, this will result in extra delay for lower priority nodes where the decoder must wait for the higher priority packets to be decoded first to perform interference cancellation which dictates parallel decoding \cite{SIC}. We study the interference alleviation scenarios for one node at a certain time slot, while other nodes' packets also are transmitted and decoded at the same time as a lower bound worse case.
\vspace{-0mm}
\subsection{Power control}\label{power_control}
The method of power control is based on increasing the SNR of node $n$ to allow it recover from the interference effect. Let $\rho_c$ be the new SNR of node $n$, while the other nodes still transmit with SNR equal to $\rho$. Then, we equate the SINR equation in (\ref{eq2.2}) to the case where no collision occurs ($N=1$) to obtain
\vspace{-0mm}
\begin{equation}\label{rhoc2}
\begin{split}
\rho_c&=\rho \ (1+\rho (N-1)).
\end{split}
\end{equation}

When a certain node transmits with SNR of $\rho_c$, its EC is the same as in the case when transmitting with SNR equals to $\rho$ while other nodes are silent. The method of power control is simple; however, it causes extra interference into other nodes due to the power increase of the recovering node.

From (\ref{rhoc2}), we define the SINR of other nodes colliding in the same network (nodes in set $s$) after the compensation of one node as 
\begin{align}\label{eq44}
	\rho_s&=\frac{\rho}{1+\rho_c+\rho (N-2)}=\frac{\rho}{1+\rho \ (\rho+1) (N-1)} 
\end{align}

Now, we are interested in comparing the per-node EC in 3 cases: \textit{i}) No collision, \textit{ii}) collision without compensation, and \textit{iii}) collision with compensation of 1 node. Here, we look for the maximum achievable per-node EC in each case. Define the collision loss factor $\alpha$ as the ratio between the maximum effective capacity $EC_{max}$ in case of collision and in case of no collision as 
\begin{equation}\label{alpha}
\alpha=\frac{EC(\rho_i,\theta,\epsilon_i^*)}{EC(\rho,\theta,\epsilon^*)} 
\end{equation}
where $\epsilon^*$ and $\epsilon_i^*$ are the optimal error probabilities for the cases of no collision ($N=1$) and collision without compensation, respectively and both are obtained from (\ref{ur5}).

To determine the effect of compensation of one node on the other nodes, we define the compensation loss factor $\alpha_{c}$ as the ratio between maximum EC of other nodes (set $s$) in case of one node compensation and in case of no compensation. That is
\begin{align}\label{alphac}
	\alpha_{c}=\frac{EC(\rho_s,\theta,\epsilon_s^*)}{EC(\rho_i,\theta,\epsilon_i^*)} 
\end{align}
where $\epsilon_s^*$ is the optimum error probability obtained from (\ref{e*}) when the SINR is set to $\rho_s$. To understand the effect of increased interference on the network performance, we study the effect of SINR variations on EC for different delay constraints.
\begin{proposition} \label{p1}
	SINR variations have comparably limited effect on EC when the delay constraint becomes more strict and vice versa.
\end{proposition}
\begin{proof}.
Please refer to Appendix B.  
\end{proof}

Furthermore, we include the compensation factor $\gamma_c$ as the ratio of the maximum achievable EC of the compensated node after and before compensation which is expressed as
\begin{equation}\label{gammac}
\gamma_c=\frac{EC(\rho,\theta,\epsilon^*)}{EC(\rho_i,\theta,\epsilon_i^*)}=\frac{1}{\alpha} 
\end{equation}
where $\gamma_c$ is a gain factor (i.e., $\gamma_c \geq 1$). Finally, we define the total loss factor $\alpha_t$ as the ratio between the maximum attainable effective capacity of colliding nodes in the system (set $s$) when a node compensates to the maximum attainable EC of these nodes if they were not colliding at all. That is
\begin{equation}\label{alphat}
\alpha_t=\frac{\alpha_c}{\gamma_c}=\alpha . \alpha_c 
\end{equation}

\subsection{Graceful degradation of the delay constraint} \label{theta comp}
Here, we discuss how to compensate for the decrease in the per-user EC for the multiuser interference scenario by changing the value of delay constraint $\theta$. More specifically, we determine how the delay exponent $\theta$ should be gracefully degraded to obtain the same $EC_{max}$ as if the target node was transmitting without collision. This represents the cases where a node has flexible QoS constraint delay wise, so that the EC could be attained given a slight variation on the overall delay as envisioned in \cite{paper1}. Let $\theta$ be the original delay exponent and $\theta_i$ represent the new gracefully degraded one; $\theta_i$ is obtained by solving
\begin{align}\label{RE}
	EC(\rho,\theta,\epsilon^*)=EC(\rho_i,\theta_i,\epsilon_i^*)
\end{align} 
where $\epsilon_i^{*}$ is the maximizer of EC for the parameters $(\rho_i,\theta_i)$ and $\epsilon_i^*$ is the optimum error probability for $(\rho_i,\theta_i)$. The solution of (\ref{RE}) renders the necessary value of $\theta_i$ to compensate for the EC decrease due to collision in this case. Notice that (\ref{RE}) can be solved numerically to obtain the necessary value of $\theta_2$ to compensate for the rate decrease due to collision in this case.

\subsection{Joint compensation model} \label{joint compensation} 
To mitigate the side effects of power control and graceful delay constraint degradation, we apply a joint model in which both methods are partially employed. Define the operational SINR in power controlled compensation for nodes in set $s$ as $\rho_{s_o}$, where $\rho_{s_o}$ lies on the interval [$\rho_s$ $\rho_i$]. Using (\ref{eq44}), the operational SNR for the recovering node can be written as
\begin{equation}\label{roco}
\rho_{c_o}=\frac{\rho}{\rho_{s_o}} -1-\rho(N-2),
\end{equation} 
and the operational point of the compensation loss factor $\alpha_{c_o}$ is
\begin{equation}\label{alphaco}
\alpha_{c_o}=\frac{EC(\rho_{s_o},\theta,\epsilon_{s_o}^*)}{EC(\rho_i,\theta,\epsilon_i^*)} 
\end{equation}
where $\epsilon_{s_o}^*$ is the optimum error probability  obtained from (\ref{e*}) for the parameters ($\rho_{s_o},\theta_1$). $\alpha_{c_o}$ is considered to be the loss factor caused by the part of compensation performed via power control. 

Next, we perform the rest of compensation via graceful degradation of $\theta$ as in Section \ref{theta comp}. To obtain $\theta_2$, we solve
\begin{align}\label{theta joint}
	&EC(\rho,\theta,\epsilon^*)=EC\left( \frac{\rho_{c_o}}{1+\rho (N-1)},\theta_2,\epsilon_2^*\right)  
\end{align}    
From (\ref{theta joint}), we compute the necessary value of $\theta_2$ to continue the compensation process via graceful degradation of the delay constraint.

Now, we propose an objective function leveraging the network performance for the joint model. First, we define the priority factor $\eta_{\alpha}$ as a measure of the risk of decrease in EC of nodes in set $s$ when the compensating node boosts its transmission power. In other words, the higher the value of $\eta_{\alpha}$, the more important it is not to allow much degradation of EC of nodes in set $s$ and hence, we try not to compensate via power control and shift compensation towards $\theta$ graceful degradation. On the other hand, we define the priority factor $\eta_\theta$ as a measure of strictness of the delay constraint (i.e., the higher the value of $\eta_\theta$, the more strict it is not to degrade delay constraint and hence, the less we are allowed to relax $\theta$ to get higher EC for the compensating user). Thus, we can formalize our objective function as the summation
\begin{equation}\label{eta}
\eta=\eta_\alpha \alpha_{c_o}+\eta_\theta \theta_2
\end{equation} 
where ($\alpha_{c_o},\theta_2$) is the operational point. Now, we choose this operational point to satisfy 
\begin{equation}\label{op}
\begin{split}
\eta_{max}=&\max_{\theta_2 \geq 0} \ \eta_\alpha \alpha_{c_o}+\eta_\theta \theta_2 \\ 
s.t \ \ & \rho_s \leq \rho_{s_o} \leq \rho_i \\
\end{split}
\end{equation} 
where the solution to this problem gives the optimum operational point which can be found from (\ref{roco}), (\ref{alphaco}) and (\ref{theta joint}).

\section{Results and discussion}
Figure 3 depicts the operational and maximum EC as a function of the number of nodes N for $T_f=1000, \rho=10$ dB, $\epsilon_t$ different values of $\theta$. The figure reveals that operating in the UR region necessitates a considerable sacrifice in the EC. Specifically, as the number of nodes increases and the SINR decreases, the amount of EC sacrifice becomes more significant as the gap between operational and optimum EC becomes higher. For example, at $N=30$, the nodes lose up to 50 \% of its EC to maintain reliability.
\begin{figure}[!t] 
	\centering
	\includegraphics[width=0.95\columnwidth]{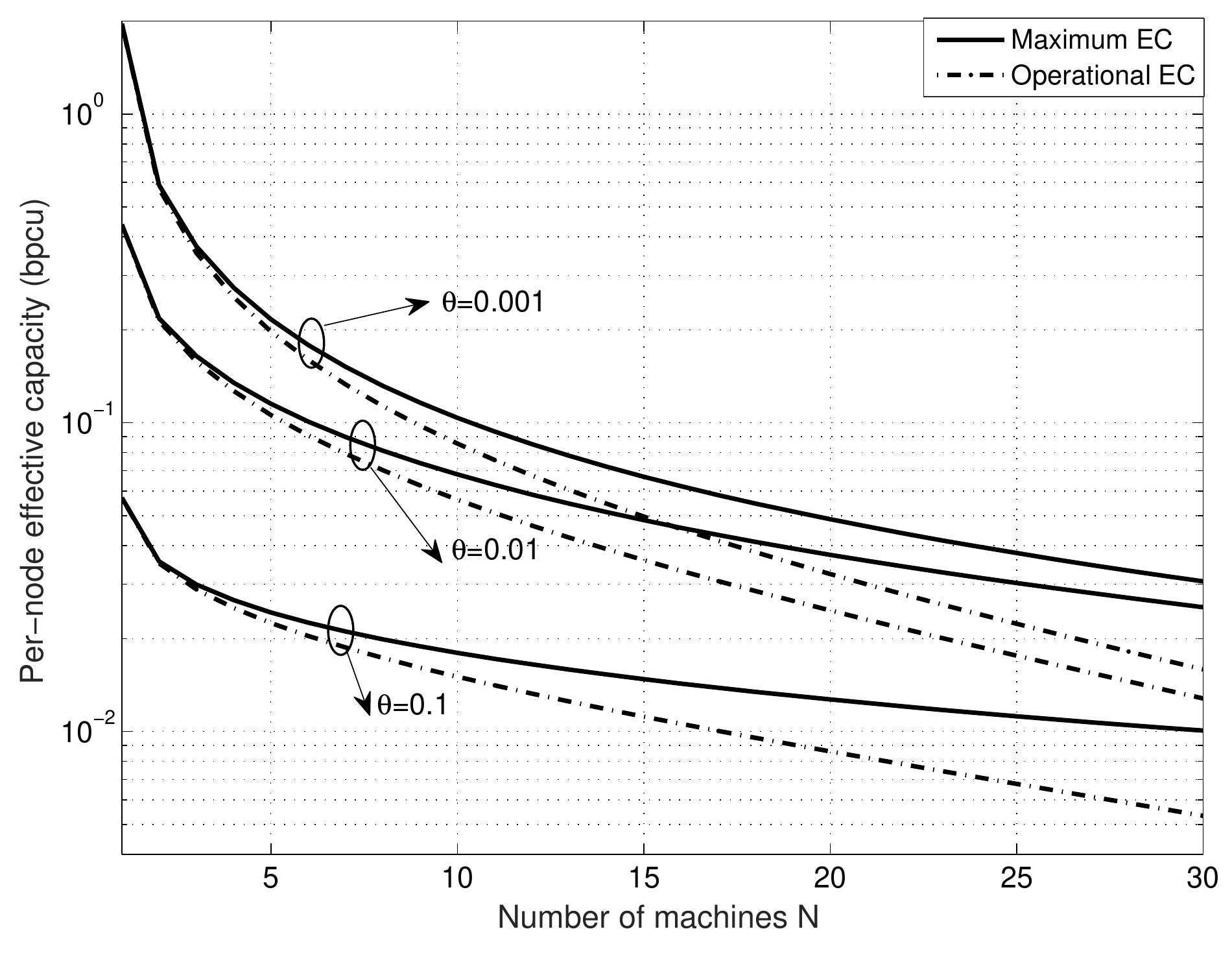}
	\vspace{-0mm}
	\caption{Operational and maximum effective capacity vs number of machines N for $T_f=1000$, and $\rho=10$.}
	\label{UR}
	\vspace{-0mm}
\end{figure}

As an example, consider 3 colliding nodes where $T_f=500, \theta=0.01$, and $\rho=0.5$. Applying (\ref{rhoc2}), we get $\rho_c=1$. Hence, the interference effect is canceled for a certain node by boosting its SNR from 0.5 to 1. To elucidate more, Figure 4 shows the effect of collision of 5 nodes with and without compensation for $T_f=1000, \theta=0.1$, and $\rho=1$. We plot the per-node EC before compensation of 1 node. Then we compare it to the effective capacities of the 4 remaining nodes after 1 node compensates using (\ref{rhoc2}). The figure also shows the EC of the compensating node.
\begin{figure}[!t] 
	\centering
	\includegraphics[width=0.95\columnwidth]{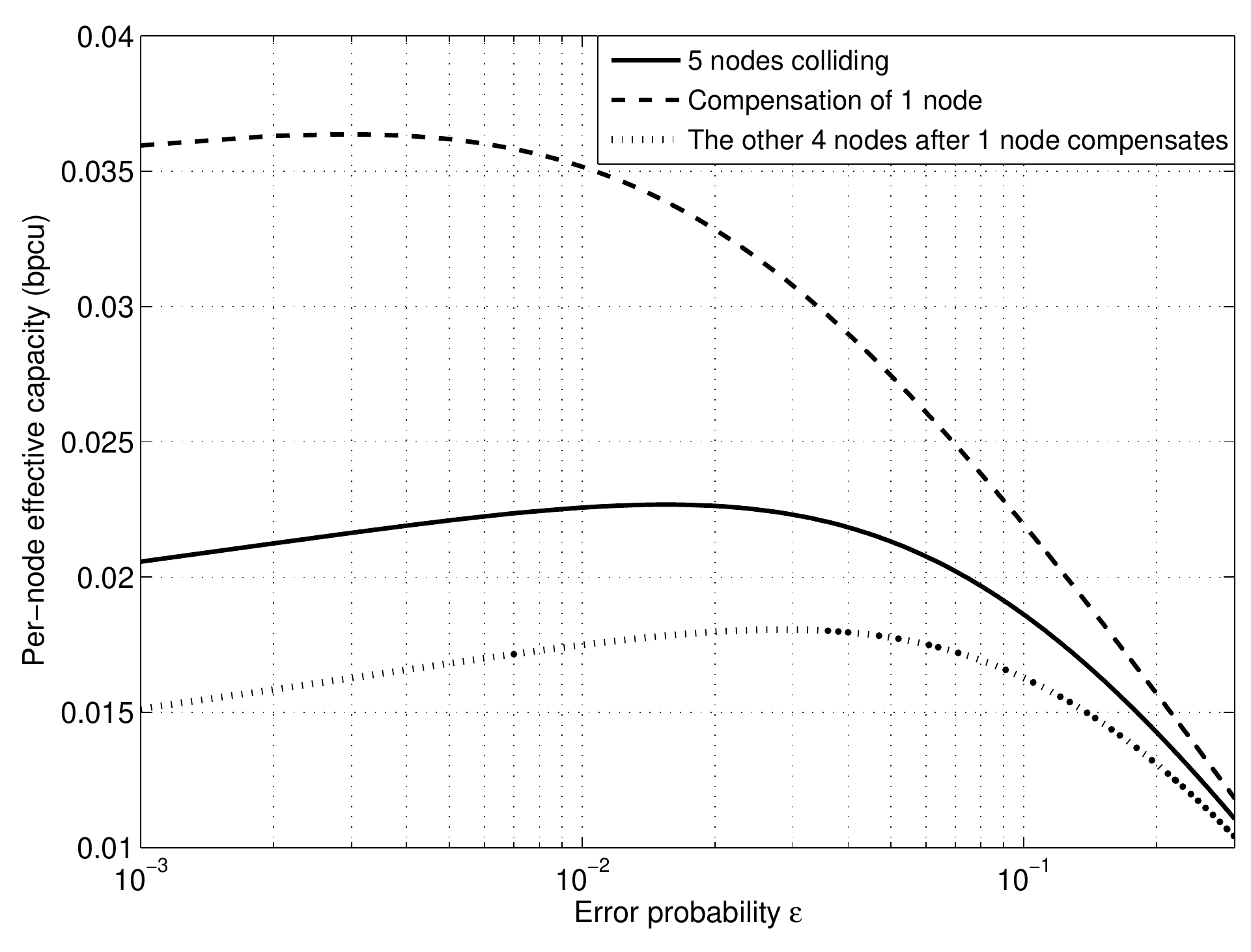}
	\centering
	\caption{Per-node effective capacity as a function of error outage probability $\epsilon$ with $T_f=1000, \theta=0.1$, and $\rho=1$}
	\label{5 nodes}
	\vspace{-0mm}
\end{figure}

Consider $T_f=1000$ and $\rho=1$, then the left axis of Figure 5 depicts the compensation loss factor $\alpha_c$ for different number of nodes $N$ with $\theta=0.1$ and 0.001.  The figure shows that $\alpha_c$ is lower for smaller values of $\theta$. Hence, the effect of compensation appears to be more severe for less stringent delay constraints. This follows from Proposition \ref{p1}, which states that SINR variations have less of an effect on delay strict networks and vice versa. Finally, we notice that the compensation loss factor decreases rapidly for a less dense network. The right axis of Figure 5 shows the compensation factor $\gamma_c$ versus the number of nodes in the system $N$. $\gamma_c$ appears to have a linear behavior as a function of $N$. That is, the effect of compensation for the compensated node increases linearly with $N$. The rate by which $\gamma_c$ increases is faster for smaller $\theta$. The compensation factor $\gamma_c$ (compensation gain) is higher for less stringent delay constraint (less $\theta$). It appears that $\alpha_c$ and $\gamma_c$ have are inversely correlated to each other as envisioned by (\ref{gammac}).  
\begin{figure}[!t] 
	\centering
	\includegraphics[width=0.95\columnwidth]{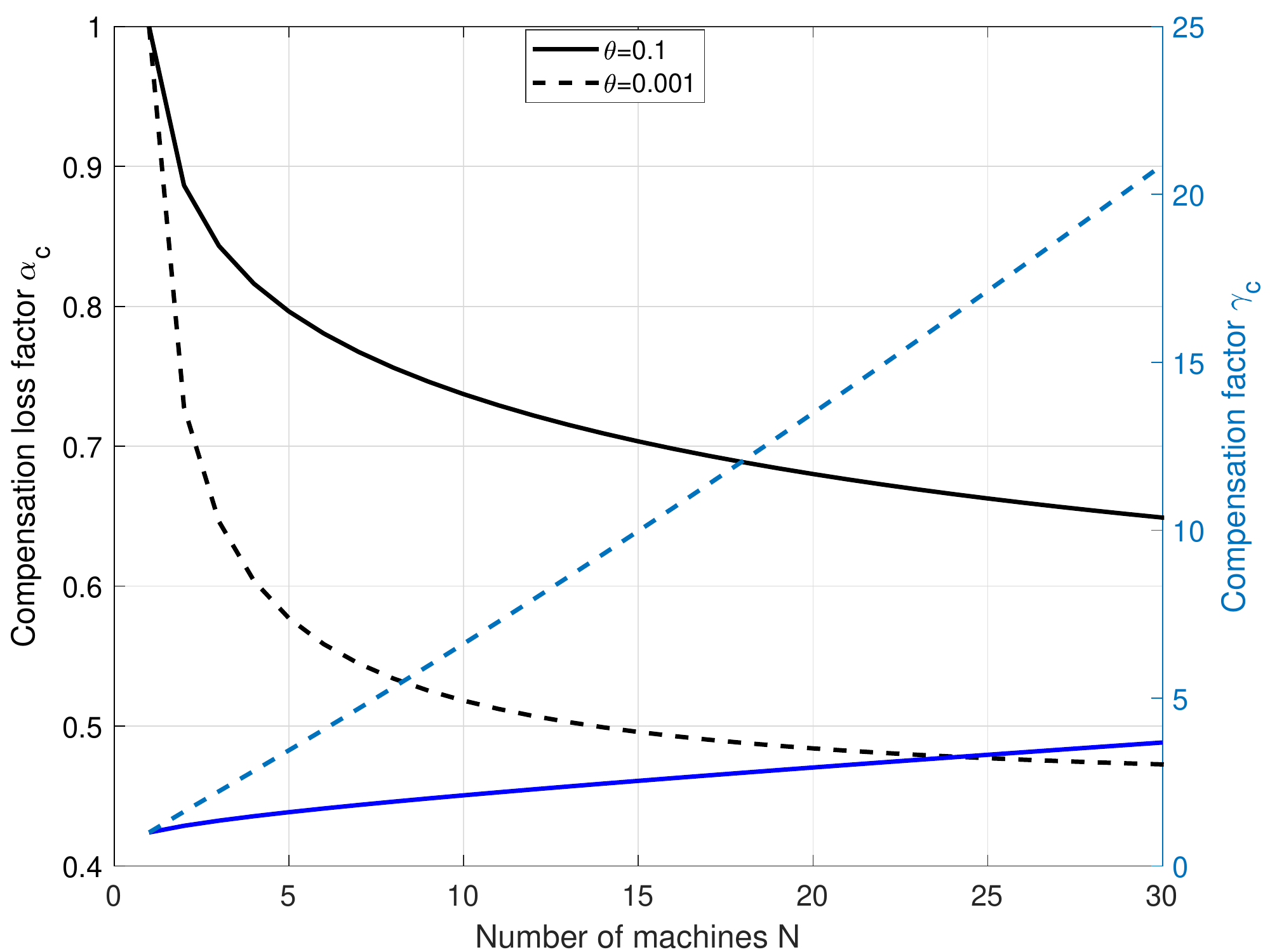}
	\centering
	\vspace{-0mm}
	\caption{Compensation loss factor $\alpha_c$ as a function of the number of nodes $N$, for distinct QoS exponents.}
	\label{alpha_c}
	\vspace{-0mm}
\end{figure}

\begin{figure}[!t] 
	\centering
	\includegraphics[width=0.95\columnwidth]{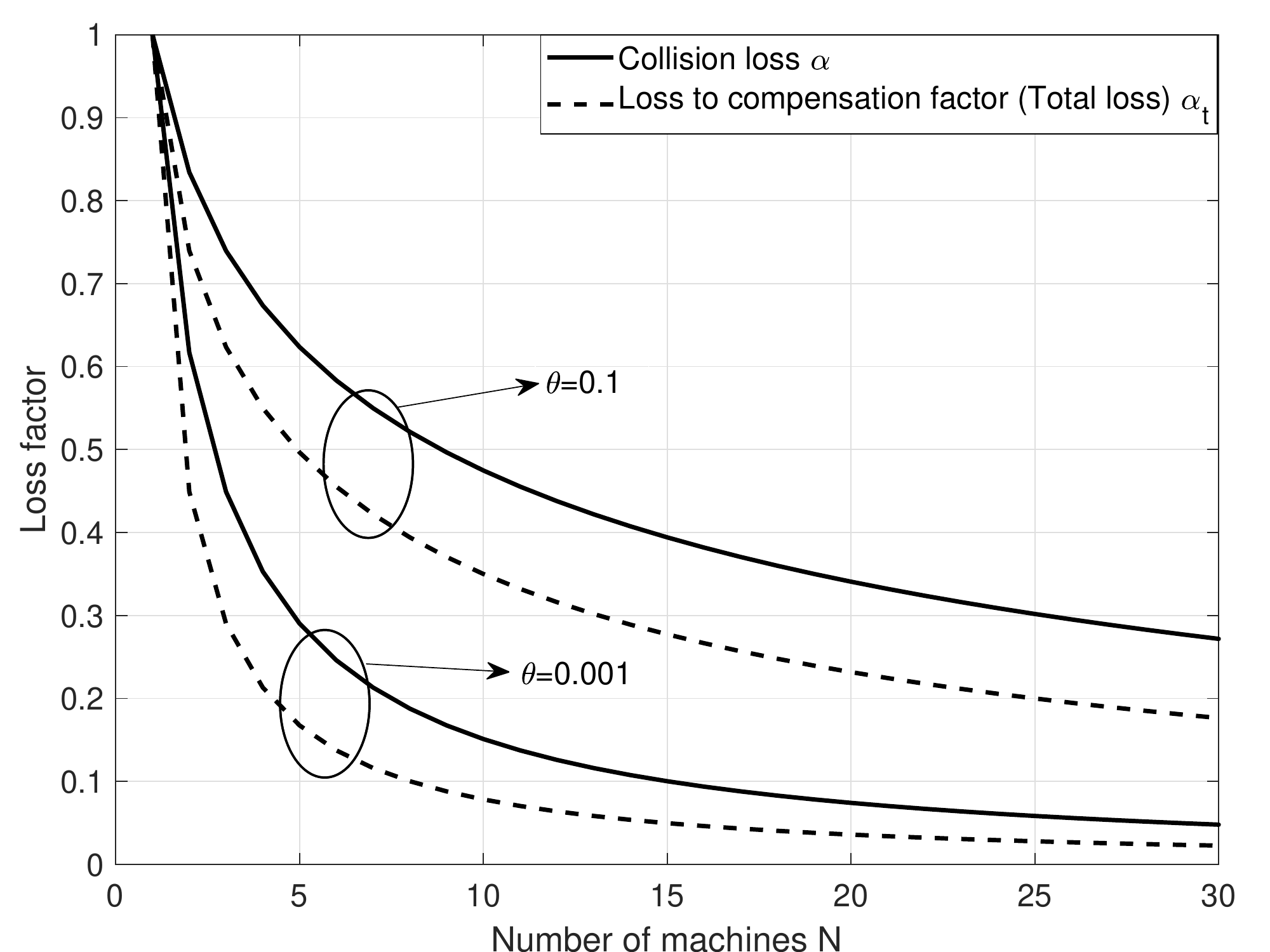}
	\centering
	\caption{Loss factors $\alpha$ and $\alpha_t$ vs number of nodes $N$}
	\label{loss}
	\vspace{-0mm}
\end{figure}

The collision loss factor $\alpha$ and the total loss factor $\alpha_t$ are also depicted in Figure 6. As observed from the figure, the total loss is nearly the same as the collision loss for small $N$. The gap starts to appear and becomes almost constant for high number of nodes. This gap is tighter in the case of small $\theta$ (i.e., the collision loss $\alpha$ is more dominant). Also, there is higher loss (both collision and total) in case of less stringent delay constraint $\theta$. Thus, the effect of collision and compensation is more annoying in case of less stringent delay constraint (smaller $\theta$). Furthermore, $\alpha$ and $\alpha_t$ have the same behavior as $\alpha_c$. They decrease rapidly for small number of nodes and tend to be constant for high $N$.

Consider (\ref{RE}) with $N=5$, $\theta_1=0.05$, $\rho=1$ and $T_f=1000$, we get $\theta_i=0.023$. Thus, by gracefully degrading the delay constraint from $0.05$ to $0.023$, we attain the same value for the maximum effective capacity $EC_{max}=0.066$ as was depicted in Figure 4 in \cite{eucnc}. Here in Figure 7, we illustrate the graceful degradation of the delay constraint by plotting the delay outage probability $P_{out\_ delay}$ as a function of the maximum delay bound $D_{max}$ before and after compensation. The figure shows that for a delay outage probability of $10^{-3}$, the compensation process is performed by extending the allowable delay $D_{max}$ from 3600 to 4600 channel uses. We perform a limited delay extension ($\approx 25 \%$) because the rise in EC partially compensates the graceful degradation of $\theta$ in (\ref{delay}). Note that the optimum error probabilities have different values in each case due to the change in SINR in (\ref{J}) \cite{eucnc}.   

\begin{figure}[!t] 
	\centering
	\includegraphics[width=0.95\columnwidth]{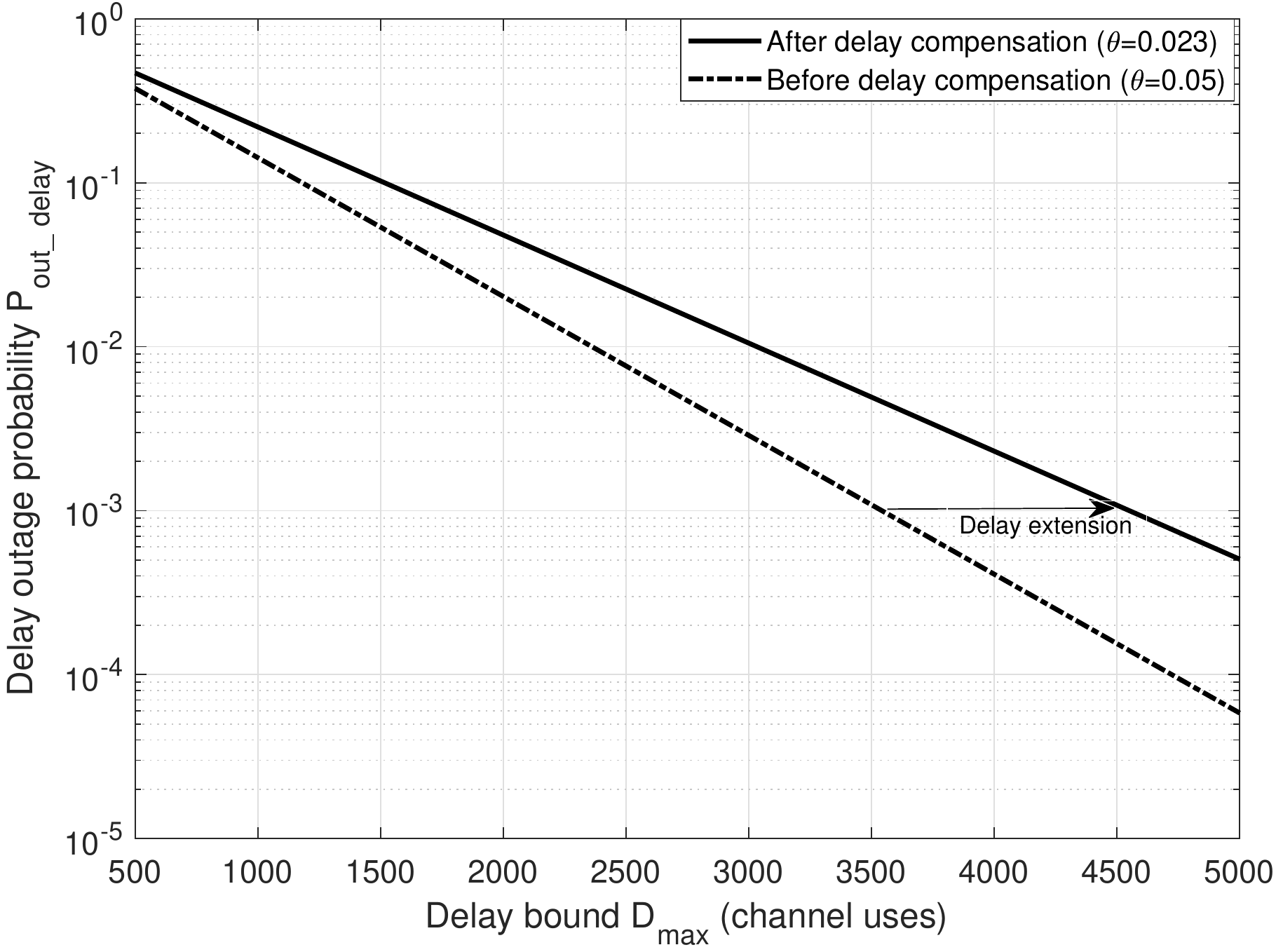}
	\centering
	\vspace{-0mm}
	\caption{Graceful degradation delay constraint $\theta$ in case of 5 nodes colliding where $T_f=1000$ and $\rho=1$.}
	\label{theta compensation}
	\vspace{-0mm}
\end{figure}

\begin{figure}[!t] 
	\centering
	\includegraphics[width=0.95\columnwidth]{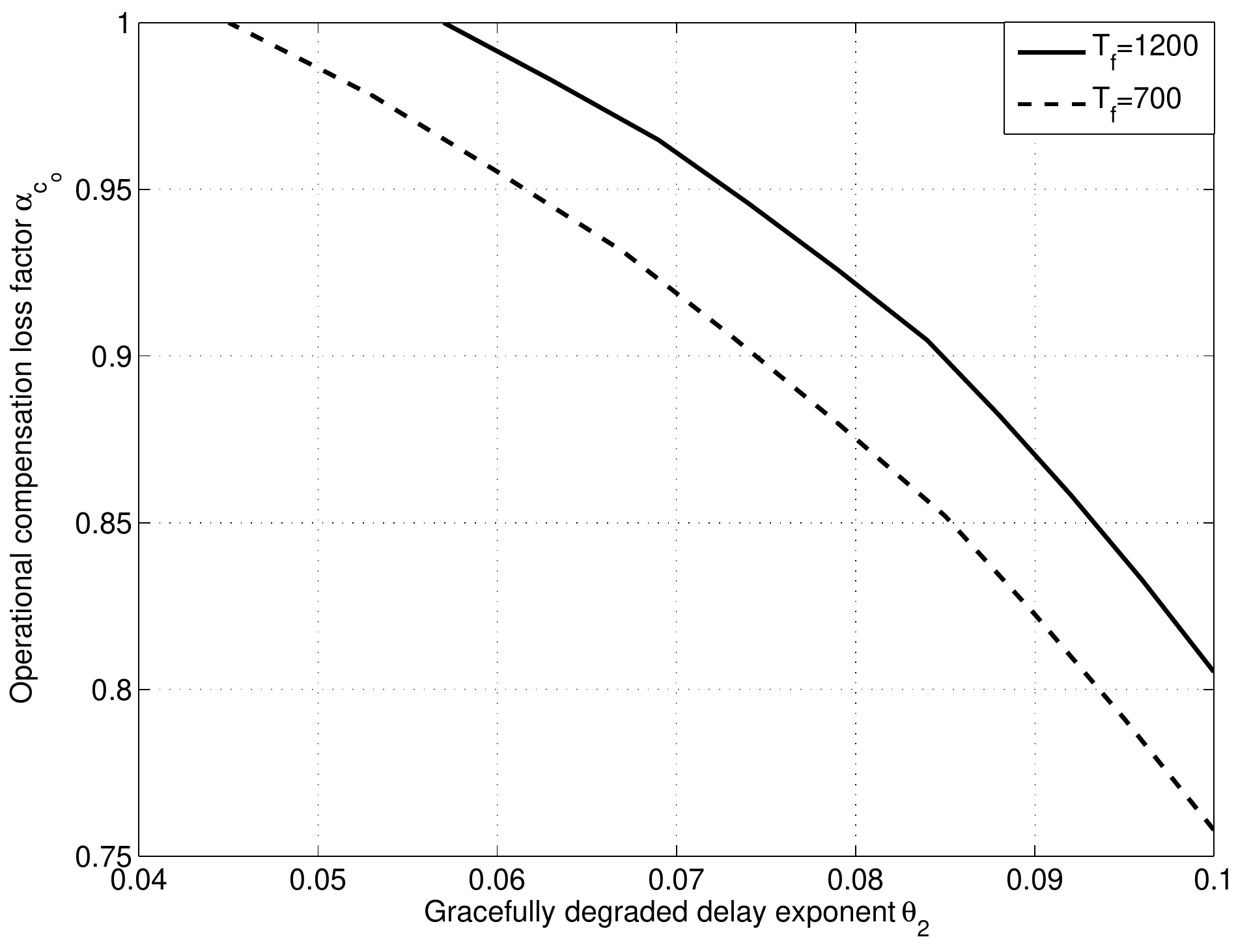}
	\centering
	\vspace{-0mm}
	\caption{Trade off between compensation loss factor via power control $\alpha_{c_o}$ and graceful degradation of delay constraint $\theta_{2}$ for different operational points}
	\label{trade off}
	\vspace{-0mm}
\end{figure}

Figure 8 illustrates different operational points for the joint model for different blocklength $T_f$ where $N=5$, $\rho=1$, $P_{out\_ delay}=10^{-3}$ and $\theta_1=0.1$. For example, when $T_f=700$, we select the operational point $\alpha_{c_o}=0.9, \theta_2=0.075$. This implies that a part of compensation will be performed via power control, which leads to $10\%$ loss in EC of other nodes (set $s$). Then, the rest of compensation will be performed by gracefully degrading its $\theta$ from 0.1 to 0.075. The maximum delay of the recovering node remains 2500 channel uses before and after recovery as restoring the EC compensates for the decrease in $\theta$ in (\ref{delay}). The figure also shows that for smaller packet sizes, the amount of losses due to compensation are higher. 

According to the system parameters, certain values of the priority factors $\eta_\alpha$ and $\eta_\theta$ may produce a concave maximization problem for the objective function $\eta$. For an MTC network with 15 devices where $T_f=1000, \rho=2, \theta_1=0.1,\eta_\alpha=1$ and $\eta_\theta=4$, the optimum value of $\rho_{s_o}$ which maximizes the objective function $\eta$ will be 0.057 according to Figure 9. This value corresponds to the operational point $\alpha_{c_o}=0.9397$ and $\theta_2=0.053$. The SNR of the recovering node becomes $\rho_{c_o}=8.08$. In other words, in order to maximize the network throughput according to the given priority factors, the compensating node boosts its SNR from 2 to 8.08 and gracefully degrades its delay exponent from 0.1 to 0.053. This results in only 6 \% loss in EC of other nodes as depicted in Figure 10. Priority factors are left for the designer's preferences depending on reliability or latency requirements.

Finally in Figure 11, we compare the EC vs the delay constraint $\theta$ for fixed and variable rate transmissions. The SNR and the blocklength are set as $\rho=1$, and $T_f=1000$, respectively. It can be observed that for high values of $\theta$, fixed rate transmission performs strictly better. The figure also confirms the fact that EC degrades with the increase of delay constraint $\theta$.

\begin{figure}[!t] 
	\centering
	\includegraphics[width=0.95\columnwidth]{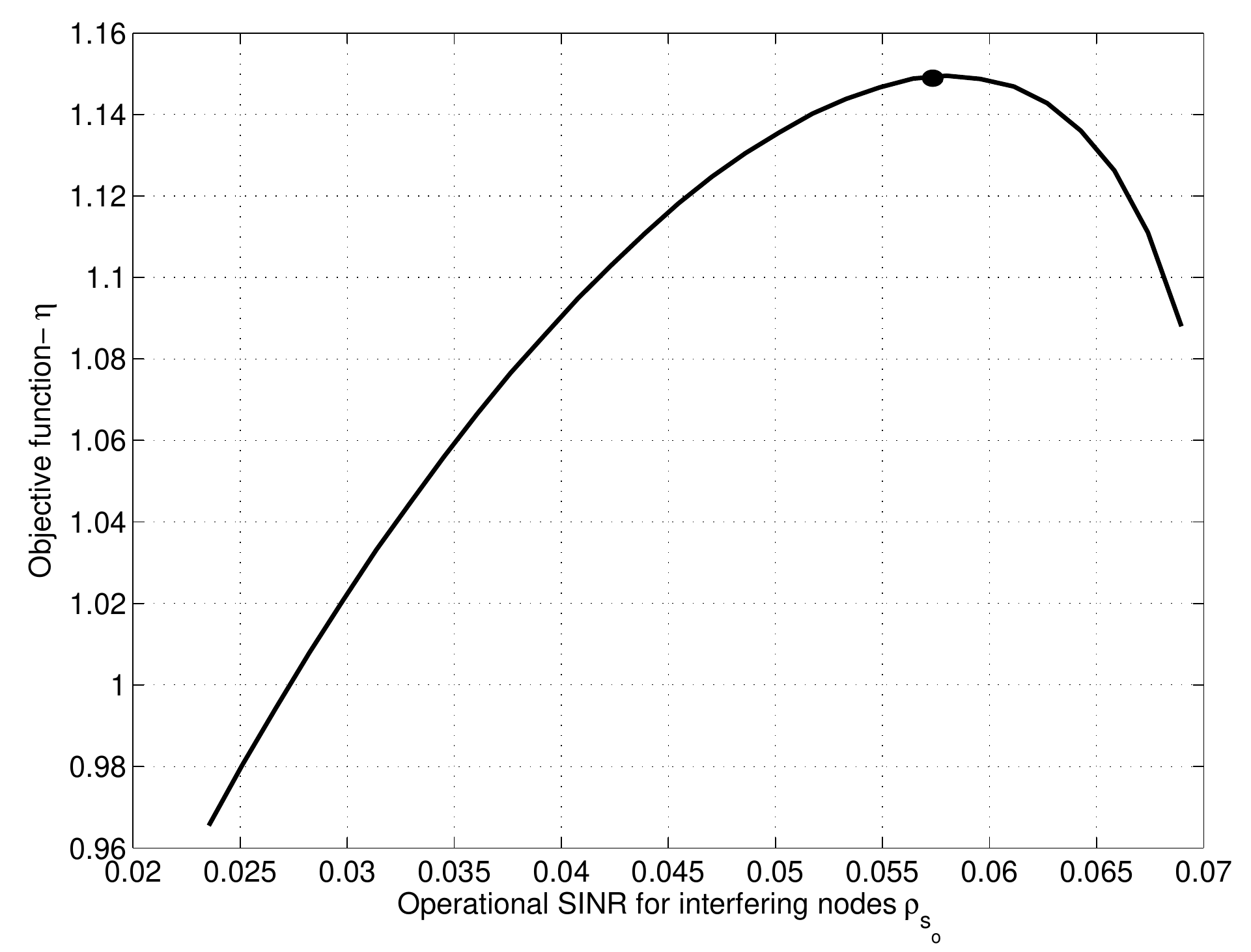}
	\centering
	\caption{Throughput $\eta$ for $T_f=1000, \theta=0.1, \rho=2$, and $N=15$}
	\label{throughput}
	\vspace{0mm}
\end{figure}

\begin{figure}[!t] 
	\centering
	\includegraphics[width=0.95\columnwidth]{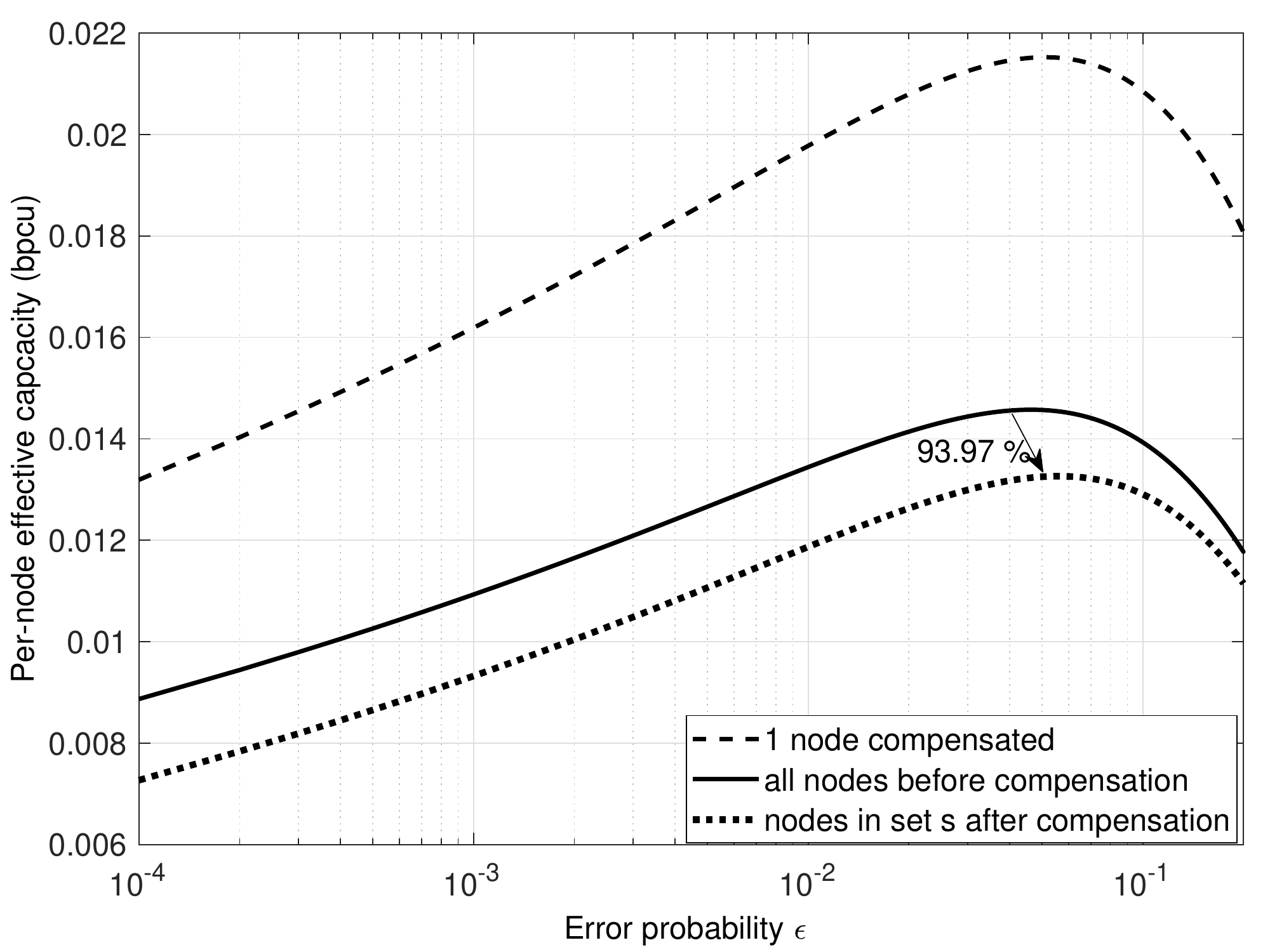}
	\centering
	\vspace{-0mm}
	\caption{Per-node effective capacity as a function of error probability $\epsilon$ before and after joint compensation for $T_f=1000, \theta_1=0.1, \rho=2$, and $N=15$}
	\label{another}
	\vspace{-0mm}
\end{figure}

 \begin{figure}[ht]
	\begin{center}
		\includegraphics*[width=1\textwidth]{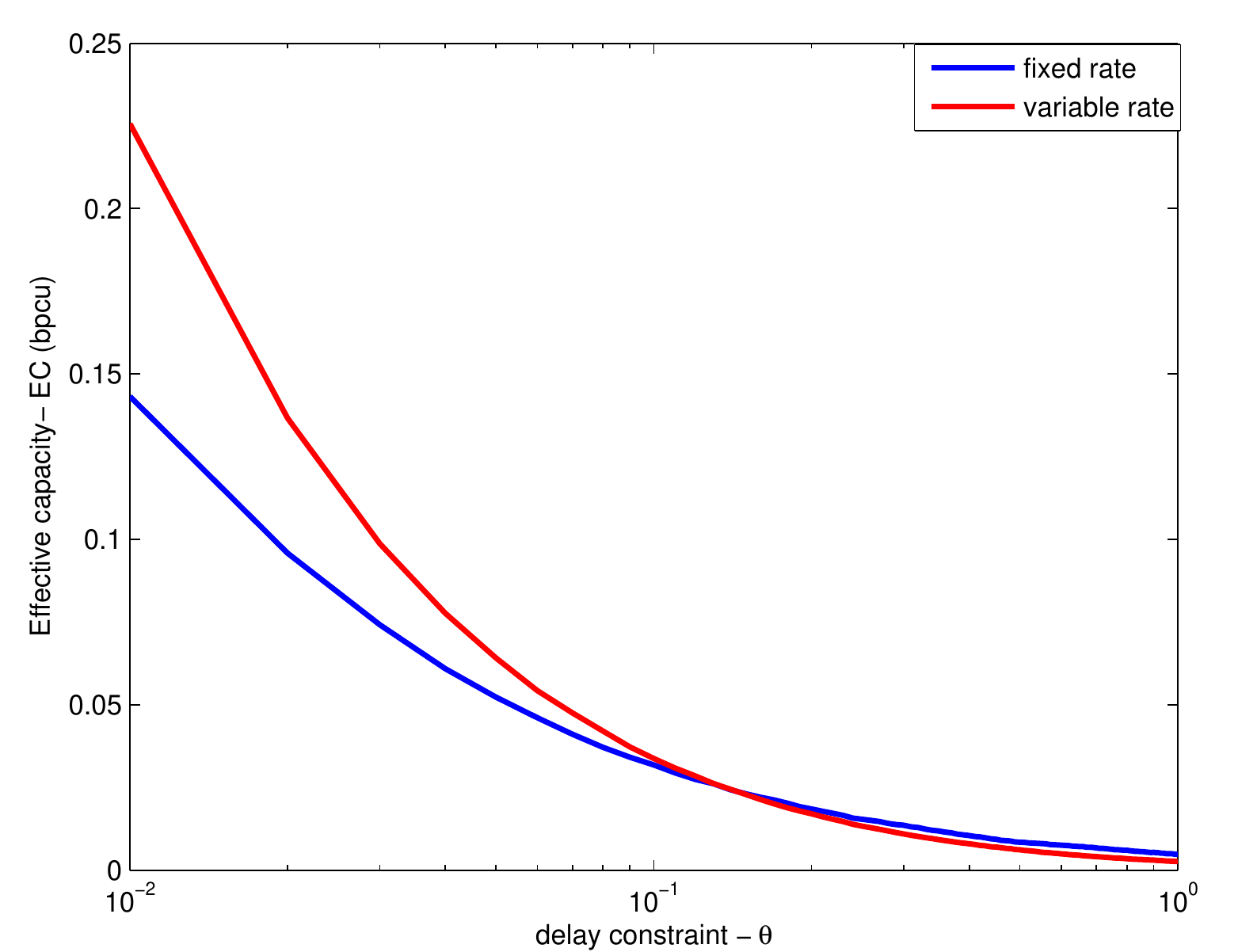}
	\end{center}
	\caption{Comaprison between fixed rate transmission and variable rate for different values of delay constraint $\theta$ where $\rho=1$, and $T_f=1000$.}
	\label{fixed rate}
\end{figure}

\section{Conclusions} \label{conclusion}

In this work, we presented a detailed analysis of the EC for delay constrained MTC networks in the finite blocklength regime. For quasi-static Rayleigh fading channels, we proposed an approximation for the EC and defined the optimum error probability. We characterized the optimization problem to maximize EC with error constraint which showed that there is a relatively small sacrifice in EC for high SINR. Our analysis indicated that SINR variations have minimum effect on EC under strict delay constraints. In a dense MTC network scenario, we illustrated the effect of interference on EC. We proposed power control as an adequate method to restore the EC in networks with less stringent delay constraints. Another method is graceful degradation of delay constraint, where we showed that a very limited extension in the delay bound could successfully recover the EC. Joint compensation emerges as a combination between these two methods, where an operational point is selected to maximize an objective function according to the networks design aspects. Finally, we concluded that for high values of $\theta$, fixed rate transmission performs strictly better. As future work, we aim to analyze the impact of imperfect CSI on the EC and coordination algorithms that maximize EC with fairness constraints.

\appendix

\section{PROOF OF LEMMA 1}
	For Rayleigh envelope of pdf given in (\ref{pdf}), the EC expression in (\ref{EC}) can be written as
\begin{align}\label{EC2}
\begin{split}
\psi(\rho_i,\theta,\epsilon)= \int_{0}^{\infty}
\left( \epsilon+(1-\epsilon)e^{-\theta T_f r}\right)  e^{-z} dz.  
\end{split}
\end{align}
From (\ref{eq3}), we have		
\begin{align}\label{e1}
\exp\left( -\theta T_f r\right) =&\exp\left(-\theta T_f \log_2(1+\rho_i z)\right) \times \notag \\ &\exp\left(\theta \sqrt{T_f(1-\frac{1}{(1+\rho_i z)^{2}})} Q^{-1}(\epsilon)\log_2e\right) .
\end{align}	
Elaborating, we attain
\begin{flalign}\label{e2}
\exp\left(-\theta T_f \log_2(1+\rho_i z)\right) 
&=(1+\rho_i z)^{d},
\end{flalign}
\begin{align}\label{e3}
\exp\left(\theta \sqrt{T_f(1-\frac{1}{(1+\rho_i z)^{2}})} Q^{-1}(\epsilon)\log_2e\right) =\exp\left(cx\right) .  
\end{align}
We resort to the second order Taylor expansion to obtain $e^{c x} = 1+c x+\frac{(c x)^2}{2}$ and place it into (\ref{e1}), then \eqref{EC2} becomes
\begin{align}\label{general2}
\begin{split}
\psi(\rho,\theta,\epsilon)&= \epsilon\int_{0}^{\infty}
e^{-z} dz+ (1-\epsilon)\left(  \int_{0}^{\infty}
(1+\rho z)^{d}e^{-z} dz  \right. \\
& +c\int_{0}^{\infty}
(1+\rho z)^{d} x e^{-z} dz +\left. \frac{c^2}{2}\int_{0}^{\infty}
(1+\rho z)^{d} x^2 e^{-z} dz \right)   
\end{split}
\end{align} 
The first integral reduces to unity. Applying Laurent's expansion for $x$ [33]\cite{Complex_Analysis}, we attain.
\begin{align}\label{x}
x\approx1-\tfrac{1}{2\left(1+\rho z \right)^{2}}.    
\end{align}	
Replacing (\ref{x}) into (\ref{general2}) yields
\begin{align}\label{general3}
\begin{split}
\psi(\rho,\theta,\epsilon)&= \epsilon+ (1-\epsilon)\left[\vphantom{\int}  e^{\frac{1}{\rho}} \rho^d\Gamma\left(d+1,\frac{1}{\rho}\right)  \right. \\
& +c\left(e^{\frac{1}{\rho}} \rho^d\Gamma\left(d+1,\frac{1}{\rho}\right) -\frac{1}{2} e^{\frac{1}{\rho}} \rho^{d-2}\Gamma\left(d-2,\frac{1}{\rho}\right)\right) \\
&+\left. \frac{c^2}{2}\left(e^{\frac{1}{\rho}} \rho^d\Gamma\left(d+1,\frac{1}{\rho}\right) - e^{\frac{1}{\rho}} \rho^{d-2}\Gamma\left(d-2,\frac{1}{\rho}\right)\right)\right].   
\end{split}
\end{align}
After manipulating (\ref{general3}) and inserting it into (\ref{EC}), we obtain (\ref{Rayleigh}). 
\section{PROOF OF PROPOSITION 1}
	Differentiating (\ref{EC}) with respect to $\rho_i$
\begin{align*}\label{}
\begin{split} 
\frac {\partial EC}{\partial \rho_i}=\frac {\partial EC}{\partial r}\frac {\partial r}{\partial \rho_i}=\frac{e^{-T_f \theta r}}{\epsilon+(1-\epsilon)e^{-T_f\theta r}} \mathcal{K},
\end{split}
\end{align*}
where $\mathcal{K}= \frac {\partial r}{\partial \rho_i} (1-\epsilon)$ is strictly positive since the achievable rate $r$ is an increasing function of the SINR $\rho_i$. Differentiating once more with respect to $\theta$
\begin{align}\label{}
\begin{split} 
\frac{\partial}{\partial \theta}\left( \frac {\partial EC}{\partial \rho_i}\right) = -\frac{\mathcal{K}T_f r e^{-T_f \theta r}}{\left(\epsilon+(1-\epsilon)e^{-T_f\theta r} \right)^2 },
\end{split}
\end{align}
which is strictly negative and thus, validating our proposition.


\section*{Declarations}
\begin{backmatter}
	
\section*{Availability of data and material}
The paper is self-contained, since we provide a mathematical framework which can be reproduced with the details provided in Sections 2 to 5, and in Section 5 numerical results and parameter settings are described in details. 
\section*{Competing interests}
The authors declare that they have no competing interests.
\section*{Funding}
This work is partially supported by Aka Project SAFE (Grant no. 303532), and by Finnish Funding Agency for Technology and Innovation (Tekes), Bittium Wireless, Keysight Technologies Finland, Kyynel, MediaTek Wireless, and Nokia Solutions and Networks.

\section*{Authors' contributions}
MS derived the equations and performed the system simulations. ED revised the equations and contributed to writing the introduction, system model and conclusion. HA supervised and reviewed the paper, while ML directed and supervised the research. All authors participated in this work, and approved the final manuscript.

\section*{Acknowledgements}
Not applicable.





\bibliographystyle{bmc-mathphys} 
\bibliography{bmc_article}





\end{backmatter}
\end{document}